%% file: 0.tex
\newif\ifEPTCS
\EPTCStrue

\ifEPTCS
\documentclass{eptcs}
\else
\documentclass[reqno]{amsart}
\input{amsFrontMatter}
\fi

\RequirePackage{hyperref}
\hypersetup{
  colorlinks=true,
  linkcolor=blue,
  citecolor=blue
}
  
\input{macros}

\ifEPTCS
\input{eptcsFrontMatter}

\fi

\begin{document}

\maketitle
\input{abstract}

\input{elements}

\input{raneytransforms}

\input{dualizing}

\input{further}

\ifEPTCS
\bibliographystyle{eptcs}
\else
\bibliographystyle{abbrv}
\fi
\bibliography{biblio}

\appendix
\end{document}

%% file: eptcsFrontMatter.tex
\author{Luigi Santocanale
  \institute{LIS, CNRS UMR 7020, Aix-Marseille Universit\'e, France}
  \email{luigi.santocanale@lis-lab.fr}
}
\def\authorrunning{L. Santocanale}
\title{Dualizing sup-preserving endomaps of a complete lattice}
\def\titlerunning{Dualizing sup-preserving endomaps of a complete lattice}

%% file: abstract.tex
\begin{abstract}
It is argued in \cite{EGHK2018} that the quantale $\EL$ of \jc
endomaps of a complete lattice $L$ is a \Gq exactly when $L$ is \cd.
We have argued in \cite{2020-RAMICS} that this \Gq structure arises
from the dual quantale 
of \mc endomaps of $L$ via \Rt{s} and extends to a \Gqoid structure on
the full subcategory of \SLatt (the category of \clatt{s} and \jc
maps) whose objects are the  \cdlatt{s}.

It is the goal of this talk to illustrate further this connection
between the quantale structure, \Rt{s}, and \cdity. \Rt{s} are indeed
\mix~maps in the \isomix category \SLatt and most of the theory can be
developed relying on naturality of these maps.
We complete then the remarks on cyclic elements of $\HLL$ developed in
\cite{2020-RAMICS} by investigating its dualizing elements. We argue
that if $\EL$ has the structure a \Fq, that is, if it has a dualizing
element, not necessarily a cyclic one, then $L$ is once more \cd. It
follows then from a general statement on \irl{s} that there is a
bijection between dualizing elements of $\EL$ and automorphisms of
$L$.  Finally, we also argue that if $L$ is finite and $\EL$ is
autodual, then $L$ is distributive.
\end{abstract}

%% file: elements.tex
\section{Lattice structure of the homsets in \SLatt}

\paragraph{The homset $\HXY$ in \SLatt.}
The category \SLatt of complete lattices and \jc functions is a
well-known \saut category
\cite{Barr1979,JoyalTierney,HiggsRowe1989,EGHK2018}. For \clatt{s}
$X,Y$, we denote by $\HXY$ the homset in this category. The
two-element Boolean algebra $\two$ is a dualizing element and
$\HomJ{X,\two}$ is, as a lattice, isomorphic to the dual lattice
$X^{op}$. More generally, the functor
$\fun{\ast} = \HomJ{\intfun,\two}$ is naturally isomorphic to the
functor $\FUN^{op}$ where, for $f : Y \rto X$,
$f^{op}: X^{op} \rto Y^{op}$ is the \radj of $f$, noted here by
$\ra{f}$ (the \ladj of an \mc function $g : Y \rto X$ shall be denoted
by $\la{g} : X \rto Y$).

Let us describe the internal structure of the homset $\HXY$ as a
\clatt.  For $x\in X$ and $y \in Y$, we define the following elements
of $\HXY$:
\begin{align*}
  c_{y}(t) & \eqdef
  \begin{cases}
    y\,, & t \neq \bot\,, \\
    \bot\,, & t = \bot\,,
  \end{cases}
  &
   a_{x}(t) & \eqdef
  \begin{cases}
    \top\,, & t \not\leq x\,,\\
    \bot\,, & t \leq x\,,
  \end{cases} \\
  (y \upTensor x)(t) & \eqdef
  \begin{cases}
    \top\,, & t \not\leq x\,, \\
    y\,, & \bot < t \leq x\,, \\
    \bot\,, & t = \bot\,,
  \end{cases}
  &
  e_{y,x}(t) & \eqdef
  \begin{cases}
    y\,, & t \not\leq x\,, \\
    \bot\,, & t \leq x\,.
  \end{cases}
\end{align*}
\begin{lemma}
  \label{lemma:tensors}
  For each $f \in \HXY$, $x \in X$, and $y \in Y$, $f(x) \leq y$ if
  and only if $f \leq y \upTensor x$. Consequently, for each
  $f \in \HXY$,
  \begin{align*}
    f & = \bigwedge \set{y \upTensor x \mid f(x) \leq y} =
    \bigwedge_{x \in X} f(x) \upTensor x\,,
  \end{align*}
  and the \jc functions of the form $y \upTensor x$ generate $\HXY$
  under arbitrary meets.
\end{lemma}
The tensor notation arises from the canonical isomorphism
$\HXY \iso (Y^{op} \otimes X)^{op}$ of \saut categories. That is,
$\HXY$ is dual to the tensor product $Y^{op} \otimes X$ and the
functions $y \upTensor x$ correspond to elementary tensors
of $Y^{op} \otimes X$.
For $f \in \HXY$, $g \in \HomJ{Y,Z}$, and $h \in \HomJ{X,Z}$, let us
recall that there exists uniquely determined maps
$g \below h \in \HXY$ and $h \upon f \in \HomJ{Y,Z}$ satisfying
\begin{align*}
  g \circ f &\leq h \Tiff f \leq g \below h \Tiff g \leq h \upon f\,.
\end{align*}
The binary operations $\below$ and $\upon$ are known under several
names: they yield left and right Kan extensions and (often when
$X = Y = Z$) they are named residuals or division operations
\cite{GJKO} or right and left implication
\cite{Rosenthal1990,EGHK2018}.  With the division operations at hand,
let us list some elementary relations between the functions previously
defined:
\begin{lemma}
  The following relations hold:
  \begin{inparaenum}[(i)]
  \item $y \upTensor x = c_{y} \vee a_{x} = c_{y} \upon c_{x}  = a_{y}
    \below a_{x}$,
  \item $e_{y,x}  = c_{y} \land a_{x} = c_{y} \circ a_{x}$,
  \item $c_{y} = y \upTensor \top = c_{y} \circ a_{\bot}$,
  \item $a_{x} = \bot \upTensor x = c_{\top} \circ a_{x}$. 
  \end{inparaenum}
\end{lemma}
\begin{proof}
  The relations $y \upTensor x = c_{y} \vee a_{x}$ and
  $e_{y,x} = c_{y} \land a_{x}$ are well known, see e.g. \cite{W1985},
  and $(iii)$ and $(iv)$ are immediate consequences of these
  relations.

  Let us focus on
  $y \upTensor x = c_{y} \upon c_{x} = a_{y} \below a_{x}$ and notice
  that, in order to make sense of these relations, we need to assume
  $c_{x} : X' \rto X$, $c_{y} : X' \rto Y$, $a_{x} : X \rto Y'$, and
  $a_{y} : Y \rto Y'$.
  Observe that $f(x) \leq y$ if and only if $f \circ c_{x} \leq c_{y}$
  and therefore, in view of Lemma~\ref{lemma:tensors} and of the
  definion of $\upon$, the relation
  $y \upTensor x = c_{y} \upon c_{x}$.
  Next, the condition $f \leq a_{y} \below a_{x}$ amounts to
  $a_{y} \circ f \leq a_{x}$, that is, for all $t \in X$, if
  $t \leq x$ then $f(t) \leq y$. Clearly this condition is equivalent
  to $f(x) \leq y$, and therefore to $f \leq y \upTensor x$.
  Finally, the relation $e_{y,x} = c_{y} \circ a_{x}$ is directly
  verified. However, let us observe the abuse of notation, since for
  the maps $a_{x}$ and $c_{y}$ to be composable we need to assume
  either $X = Y$ or $a_{x} : X \rto \two$ and $c_{y} : \two \rto Y$.
\end{proof}

\paragraph{\Mc functions as tensor product.} 
Let $\HMXY$ denote the poset of \mc functions from $X$ to $Y$, with
the pointwise ordering. Observe that, as a set, $\HMXY$ equals
$ \HomJ{X^{op},Y^{op}}$. Yet, as a poset or a lattice, the equality
$\HMXY = \HomJ{X^{op},Y^{op}}^{op}$ is the correct one. As a matter of
fact, we have $f \leq_{\HMXY} g$ iff $f(x) \leq_{Y} g(x)$, all
$x \in X$, iff $g(x) \leq_{Y^{op}} f(x)$, all $x \in X$, iff
$g \leq_{\HomJ{X^{op},Y^{op}}} f$. Using standard isomorphisms of
\saut categories, we have
\begin{align*}
  \HMXY & = \HomJ{X^{op},Y^{op}}^{op} \iso Y \tensor X^{op}\,.
\end{align*}
That is, the set of \mc functions from $X$ to $Y$ can be taken as a
concrete realization of the tensor product $Y \tensor X^{op}$. This
should not come as a surprise, since it is well-known that the set of
Galois connections from $X$ to $Y$---that is, pairs of functions
$(f : X \rto Y,g : Y \rto X)$ such that $y \leq f(x)$ \tiff
$x \leq g(y)$---realizes the tensor product $Y \otimes X$ in \SLatt,
see e.g. \cite{Shmuely1974,Nelson1976} or \cite[\S 2.1.2]{EGHK2018}. Such a pair of
functions is uniquely determined by its first element, which is an \mc
functions from $X^{op}$ to $Y$.

Notice now that
\begin{align}
  \label{eq:radj}
  \HXY^{op} & \iso Y^{op} \tensor X \iso X \tensor Y^{op} \iso \HMYX\,,
\end{align}
from which we derive the following principle:
\begin{fact}
  \label{fact:bijection}
  There is a bijection beweeen \jcf{s} from $\HXY^{op}$ to $\HXY$ and
  \jcf{s} from $\HMYX$ to $\HXY$.
\end{fact}
The map yielding the isomorphism in equation \eqref{eq:radj} is
$\rho$, the operation of taking the \radj. The bijection stated in
Fact~\ref{fact:bijection} is therefore obtained by precomposing with
$\rho$.

We exploit now the work done for $\HXY$ to recap the structure of
$\HMXY$ as a tensor product. Consider the maps
\begin{align*}
  \gamma_{y}(t) & \eqdef
  \begin{cases}
    \top \,, & t = \top\,, \\
    y \,, & \text{otherwise}\,,
  \end{cases}
  &
  \alpha_{x}(t) & \eqdef
  \begin{cases}
    \top\,, & x \leq t,\\
    \bot\,, & \text{otherwise}\,,
  \end{cases}
  \\
  y \doTensor x(t) & =
  \begin{cases}
    \top\,, & t = \top\,,\\
    y \,, & x \leq t \,,\\
    \bot\,, & \text{otherwise}\,.
  \end{cases}
\end{align*}
By dualizing Lemma~\ref{lemma:tensors}, we observe that the relation
$y \doTensor x = \gamma_{y} \land \alpha_{x}$ holds, the maps
$y \doTensor x$ realize the elementary tensors of the (abstract)
tensor product $Y \otimes X^{op}$, $\HMXY$ is join-generated by these
maps, and every $g \in \HMXY$ can be canonically written as
$g = \bigvee_{x \in X} g(x)\doTensor x$.

\smallskip

Recall that a \bimorphism $\psi : Y \times X^{op} \rto Z$ is a
function that is \jc in each variable, separately. This in particular
means that \meets in $X$ are transformed into \joins in $Z$. The
universal property of $\HMXY$ as a tensor product can be therefore
stated as follows:
\begin{fact}
  Given a \bimorphism $\psi : Y \times X^{op} \rto Z$, there exists a
  unique \jc functions $\tilde{\psi} : \HMXY \rto Z$ such that
  $\tilde{\psi}(y \doTensor x) = \psi(y,x)$. For $g \in \HMXY$,
  $\tilde{\psi}(g)$ is defined by
  \begin{align*}
    \tilde{\psi}(g) & \eqdef \bigvee_{x \in L} \psi(g(x),x)\,.
  \end{align*}
\end{fact}

%% file: raneytransforms.tex
\section{Raney's transforms}

For $g \in \HMXY$ and $f \in \HXY$, define
\begin{align*}
  \joinof{g}(x) & \eqdef \bigvee_{x \not\leq t} g(t)\,,
  &
  \meetof{f}(x) & \eqdef \bigwedge_{t \not\leq x} f(t)\,.
\end{align*}
It is easily seen that $\joinof{g}$ has a right adjoint, so
$\joinof{g} \in \HXY$, and that $\meetof{f}$ has a \ladj, so
$\meetof{f}$ belongs to $\HMXY$.  We call the operations
$\joinof{\FUN}$ and $\meetof{\FUN}$ the \Rt{s}, even if Raney defined
these transforms on Galois connections. (In \cite{2020-RAMICS} we
explicitly related these maps to Raney's original way of defining
them). Notice that $\joinof{g} \leq f$ if and only if
$g \leq \meetof{f}$, so $\meetof{\FUN}$ is \radj to $\joinof{\FUN}$.
\Rt{s} have been the key ingredient allowing us to prove in \cite{CWO}
that $\HomJ{C,C}$ is a \Gq if $C$ is a complete chain and, lately in
\cite{2020-RAMICS}, that the full-subcategory of \SLatt whose objects
are the \cdlatt{s} is a \Gqoid.

\smallskip

Consider the \bimorphism $e : Y \times X^{op} \rto \HXY$ sending $y,x$
to $e_{y,x} = c_{y} \circ a_{x} \in \HXY$ and its extension
\begin{align*}
  \tilde{e}(f) & = \bigvee \set{c_{f(t)} \circ a_{t} \mid t \in X}\,.
\end{align*}
By evaluating $\tilde{e}(f)$ at $x \in L$,
we obtain
\begin{align*}
  \tilde{e}(f)(x) & = \bigvee \set{(c_{f(t)} \circ a_{t})(x) \mid t \in L}
  =
  \bigvee_{x \not\leq t} f(t) = \joinof{f}(x)\,,
\end{align*}
that is, $\tilde{e}(f) = \joinof{f}$.
Remark now that $e : Y \times X^{ op} \rto \HXY$ is the
(set-theoretic) transpose of the trimorphism
\begin{align*}
  \langle y,x,t\rangle & =
  \begin{cases}
    \bot\,, & t \leq x \\
    y \,, & \toth\,.
  \end{cases}
\end{align*}
Consequently, \Rt $\joinof{\FUN}$ is the transpose of the map
\begin{align}
  \label{eq:downTransposeMix}
  Y \otimes X^{op} \tensor X \rto[\iso] Y \otimes \HomJ{X,\two}
  \tensor  X \rto[Y \tensor \,eval\,] Y \otimes 2 \rto[\iso] Y \,.
\end{align}
In the category \SLatt,  
$\two$ is both the unit for the tensor product $Y\tensor X$ and its
dual $(Y^{op} \tensor X^{op})^{op} \iso \HomJ{X,Y^{op}}$.  \saut
categories with this property are examples of \isomix~categories in
sense of \cite{BCST1996,CS1997,BCS2000} where the transpose of the map
in \eqref{eq:downTransposeMix} is named $\mix$.  That \Rt{s} are
$\mix$ maps was recognized in \cite{HiggsRowe1989} where also the
nuclear objects---i.e. those objects whose \mix~maps are
invertible---in the category $\SLatt$ were characterized (using \RT)
as the \cd lattices. The importance of this characterization stems
from the fact that the nucleus of a \SMC category---that is, the full
subcategory of nuclear objects---yields a \radj to the forgetful
functor from the category of compact closed categories to that of \SMC
ones, as mentioned in \cite{Rowe1988}. In particular, the full
subcategory of \SLatt whose objects are the \cdlatt{s} is more than
\SMC or \saut, it is compact closed \cite{KellyLaplaza1980}.

A key property of \Rt{s}, importantly used in \cite{CWO,2020-RAMICS},
is the following. For $g \in \HMXY$ and $f \in \HXY$, the relations
\begin{align}
  \label{eq:commRadjLadj}
 \ra{\joinof{g}} & = \meetof{\la{g}}\,, & \la{\meetof{f}} & = \joinof{\ra{f}}\,,  
\end{align}
hold.  This property might be directly verified, as we did in
\cite{CWO,2020-RAMICS}. It might also be inferred from the
commutativity of each square in the diagram below:
\begin{center}
  \begin{tikzcd}
    \HMXY\ar[rrrr,"\joinof{\FUN}_{X,Y}"] \ar[rrd,equal] \ar[ddd,"\ell"']
    &&& &\HXY
    \ar[ddd,"\rho"]\\
    & &Y \tensor X^{op} \ar[d,"\sigma"]\ar[rru,"\mix_{Y,X}"']&& \\
    & & X^{op} \tensor Y
    \ar[rrd,"\mix_{X^{op},Y^{op}}"]&&\\
    \HomM{Y^{op},X^{op}} \ar[rru,equal]\ar[rrrr,"\joinof{\FUN}_{Y^{op},X^{op}}"]
    \ar[dd,equal]
    &&& & \HomJ{Y^{op},X^{op}}  \ar[dd,equal]\\
    &&&&\\
    \HomJ{Y,X}^{op}\ar[rrrr,"\meetof{\FUN}_{Y,X}"]&&&& \HomM{Y,X}^{op}
  \end{tikzcd}
\end{center}

\paragraph{Naturality of Raney's transforms.}
Observe now that, for $g : X' \rto X$ and $f : Y \rto Y'$, we have
\begin{align*}
  f \circ c_{y} & = c_{f(y)} \,, & a_{x} \circ g & = a_{\ra{g}(x)}\,,
  \quad\tand\quad
  f \circ e_{y,x} \circ g  =
  e_{f(y),\ra{g}(x)}\,,
\end{align*}
implying that the following diagram commutes:
\begin{center}
  \begin{tikzcd}
    Y \tensor X^{op} \ar[d,"f \tensor g^{op}"]\ar[rr,"\joinof{\FUN}_{Y,X}"]&& \HXY \ar[d,"\HomJ{g,f}"]\\
    Y' \tensor X'{}^{op} \ar[rr,"\joinof{\FUN}_{Y',X'}"]& &\HomJ{X',Y'}
  \end{tikzcd}
\end{center}
That is, Raney's transform $\joinof{\FUN}$ is natural in both its
variables.
Let us remark on the way the following:
\begin{proposition}
  There are exactly two natural arrows from $Y \tensor X^{op}$ to
  $\HXY$, the trivial one and Raney's transform.
\end{proposition}
In order to simplify reading, we use $\psi$ both for a \bimorphism
$\psi : Y \times X^{op} \rto Z$ and for its extension to the tensor
product $\tilde{\psi} : Y \tensor X^{ op} \rto Z$.
\begin{proof}
  If $\psi$ is natural, then
  \begin{align*}
    \psi(y,x) & = \psi(c_{y}(\top),\gamma_{x}(\bot))
    = c_{y} \circ \psi(\top,\bot) \circ a_{x}\,,
  \end{align*}
  since $\ra{a_{x}} = \gamma_{x}$.  Let $f = \psi(\top,\bot)$.  If
  $f = \bot$, then $\psi$ is the trivial map.  Otherwise,
  $f \neq c_{\bot}$ and $f(\top) \neq \bot$.  Then, observing that
  $f \circ a_{x} = c_{f(\top)} \circ a_{x}$ and that
  $c_{y} \circ c_{z} = c_{y}$ for $z \neq \bot$, it follows that
  \begin{align*}
    \psi(y,x) & = c_{y} \circ f \circ a_{x} = c_{y} \circ c_{f(\top)}
    \circ a_{x} = c_{y} \circ a_{x}\,.
    \tag*{\qedhere}
  \end{align*}
\end{proof}
\begin{remark}
  Similar considerations can be developed if naturality is required in
  just one variable. For example, if the \bimorphism
  $\psi : Y \times X^{op} \rto \HXY$ is such that
  $\psi(y,x) \circ g = \psi(y, \ra{g}(x))$, then
  $\psi(y,x) = \chi(y) \circ a_{x}$ for some
  $\chi : Y \rto \HomJ{X,Y}$.
  \EndOfRemark
\end{remark}

For $f : X \rto Y$ in the category $\SL$, let $j = \ra{f} \circ f$ and
$o = f \circ \ra{f}$. Denote by $X_{j}$ (resp. $Y_{o}$) the set of
fixed points of $j$ (resp., of $o$). Then, we have a standard
(epi,iso,mono)-factorization
\begin{center}
  \begin{tikzcd}
    X \ar[r,"f"] \ar[d,"j"] & Y \\
    X_{j} \ar[r,"\simeq"] & Y_{o} \ar[u,hookrightarrow]
  \end{tikzcd}
\end{center}
Thus, $f$ is mono if and only $j = id_{X}$ and $f$ is epic if and only
if $o = id_{Y}$.
Notice that $Y_{o}$ is the image of $X$ under $f$, while $X_{j}$ is
the image of $Y$ under $\ra{f}$.
We apply this factorization to \Rt{s}.
\begin{definition}
  A \jcf $f : X \rto Y$ is \emph{tight} if $\jm{f} = f$, or,
  equivalently, if it belongs to the image of $\HMXY$ via the Raney's
  transform $\joinof{\FUN}$. We let $\HRXY$ be the set of tight
  functions from $X$ to $Y$.
\end{definition}

By its definition, $\HRXY$ is the sub-join-semilattice of $\HXY$
generated by the $c_{y} \circ a_{x}$. Moreover, it is easily seen that
$w \upTensor z \in \HRXY$, for each $w \in Y$ and $z \in X$, and that
$c_{y} \circ a_{x} \leq w \upTensor z$ if and only of $y \leq w$ or
$z \leq x$. From these relations, $\HRXY$ yields a concrete
representation of Wille's tensor product $Y \widehat{\tensor} X^{op}$,
see \cite{W1985}, which, for finite lattices, coincides with the Box
tensor product of \cite{GW1999}.

\medskip

Next, we list some immediate consequences of naturality of \Rt{s}:
\begin{proposition}
  \label{prop:ConsNat}
  The following statements hold:
  \begin{enumerate}[(i)]
  \item $\HRXY$ is a bi-ideal of $\HXY$.
  \item For a \clatt $L$, the transform
    $\funJoinOf : \HomM{L,L} \rto \HomJ{L,L}$ is surjective if and
    only if $id_{L}\in \HomR{L,L}$, that is, if $id = \jm{id}$.
  \item For each \clatt $L$, the pair $(\HomR{L,L},\circ)$ is a
    quantale.
\end{enumerate}
\end{proposition}
\begin{proofExtended}
  \begin{inparaenum}[(i)]
  \item 
  Recall that $h$ belongs to the image of $M \tensor L^{op}$ if and
  only if $h = \jm{h}$. Thus, if $h$ has this property and
  $f : M \rto M$, then by naturality
  $f \circ h = f \circ (\jm{h}) = \HomJ{L,f}(\jm{h}) = \joinof{((f
    \tensor L^{op})(\meetof{h}))}$, showing that $ f \circ h$ belongs
  to the image as well.
  Similarly, if $g : L \rto L$, then
  $h \circ g= (\jm{h}) \circ g = \HomJ{g,M}(\jm{h}) = \joinof{((M
    \tensor g^{op})(\meetof{h}))}$.  
\item 
  If $id_{L}\in \HomR{L,L}$, then, for each $f \in \Hom{L,L}$,
  $f = f \circ id \in \HomR{L,L}$, since $\HomR{L,L}$ is a bi-ideal.
  \end{inparaenum}
\end{proofExtended}
Let us recall that \RT \cite{Raney60} characterizes \cdlatt{s} as
those complete lattices satisfying the identity
\begin{align*}
  z & = \bigvee_{z \not\leq x} \bigwedge_{y \not \leq x} y \,.
\end{align*}
This identity is exactly the identity $id = \jm{id}$ or, as we have
seen in Proposition~\ref{prop:ConsNat}, the identity $f = \jm{f}$
holding for each $f \in \HLL$. Since \cdity is autodual (at least in a
classical context), we derive that \Rt
$\funJoinOf : \HomM{L,L} \rto \HomJ{L,L}$ is surjective if and only if
it is injective.

\medskip

We conclude this section with a glance at the quantale $(\HRLL,\circ)$
of tight maps, where $L$ is an arbitrary \clatt $L$. We pause before
for a technical lemma needed end the section and later on as well.
Recalling the equations in \eqref{eq:commRadjLadj}, let us define
\begin{align*}
  \Star{f} & \eqdef \la{\meetof{f}} \quad (\,= \joinof{\ra{f}} \,)\,,
\end{align*}
and observe the following:
\begin{lemma}
  \label{lemma:relationDual}
  For each $x \in X$, $y \in Y$, and $f \in \HXY$, the following
  conditions are equivalent:
  \begin{inparaenum}[(i)]
  \item for all $t \in X$, $x \leq t$ or $y \leq f(t)$,
  \item $c_{y} \circ a_{x} \leq f$
  \item $y \doTensor x \leq \meetof{f}$
  \item 
    $y \leq \meetof{f}(x)$
  \item $\Starf{y} \leq x$
  \item $\Star{f} \leq x \upTensor y$.
  \end{inparaenum}
\end{lemma}
\begin{proof}
  \begin{inparadesc}
  \item[\stepEquiv{i}{ii}] direct verification.
  \item[\stepEquiv{ii}{iii}] since
    $c_{y} \circ a_{x} = \joinof{(y \doTensor x)}$,
    $c_{y} \circ a_{x} \leq \meetof{f}$ and by the adjunction
    $\joinof{\FUN} \adj \meetof{\FUN}$.
  \item[\stepEquiv{iii}{iv}] by the dual of Lemma~\ref{lemma:tensors}.
  \item[\stepEquiv{iv}{v}] since $\Star{f} \adj \meetof{f}$.
  \item[\stepEquiv{iv}{v}] by Lemma~\ref{lemma:tensors}.
  \end{inparadesc}
\end{proof}

\begin{proposition}
  \label{prop:QTightNotUnital}
  Unless $L$ is a \cdlatt (in which case $\HomR{L,L} = \Hom{L,L}$),
  the quantale $(\HomR{L,L},\circ)$ is not unital.
\end{proposition}
\begin{proof}
  Let $u$ be unit for $(\HomR{L,L},\circ)$ and write
  $u = \bigvee_{i \in I} c_{y_{i}} \circ a_{x_{i}}$.  For an arbitrary
  $x \in L$, evaluate at $x$ the identity
  \begin{align*}
    a_{x} & = a_{x} \circ u = \bigvee a_{x}\circ c_{y_{i}} \circ a_{x_{i}}
  \end{align*}
   and deduce that, for each $i \in I$,
  $\bot = a_{x}\circ c_{y_{i}} \circ a_{x_{i}}(x)$.
  This happens exactly when $x \leq x_{i}$ or $y_{i} \leq x$, that is,
  when $c_{y_{i}}\circ a_{x_{i}} \leq x \upTensor x$.  Since $x \in L$
  and $i \in I$ are arbitrary, we have, within $\Hom{L,L}$,
  $u  = \bigvee_{i \in I} c_{y_{i}}\circ a_{x_{i}} \leq
  \bigwedge_{x \in L} x \upTensor x = id_{L}$.
  Again, for $y \in L$ arbitrary, evaluate at $\top$ the identity
  \begin{align*}
    c_{y} & = u \circ c_{y}  = \bigvee  c_{y_{i}} \circ a_{x_{i}}\circ c_{y}\,.
  \end{align*}
   and deduce that $y = \bigvee_{y \not\leq x_{i}} y_{i}$.
  Considering that $c_{y_{i}} \circ a_{x_{i}} \leq id_{L}$, then we have
  $y_{i} \leq \meetof{id}(x_{i})$ and therefore
    \begin{align*}
    y & = \bigvee_{y \not\leq x_{i}} y_{i} \leq \bigvee_{y \not\leq
      x_{i}} \meetof{id}(x_{i}) \leq \bigvee_{y \not\leq t}
    \meetof{id}(t) = \jm{id}(y)\,.
  \end{align*}
  Since this holds for any $y \in L$, $id \leq \jm{id}$ and since the
  opposite inclusion always holds, then $id = \jm{id}$. By \RT, $L$ is
  a \cdlatt.  
\end{proof}

Recall that a \emph{dualizing element} in a quantale $(Q,\circ)$ is an
element $0 \in Q$ such that
$0 \upon (x \below 0) = (0 \upon x) \below 0 = x$, for each $x \in Q$.
As consequences of Proposition~\ref{prop:QTightNotUnital}, we obtain:
\begin{corollary}
  Unless $L$ is a \cdlatt,
  \begin{enumerate}[(i)]
  \item the quantale $(\HomR{L,L},\circ)$ has no dualizing element,
  \item
    the interior operator $\interior{\FUN}$ obtained
  by composing the two Raney's transform is not a conucleus on
  $\HLL$.
  \end{enumerate}

\end{corollary}
\begin{proof}
  \begin{inparaenum}[(i)]
  \item If $0$ is dualizing, then $0 \below 0$ is a unit of the
quantale. 
\item 
  We argue that the inclusion
  $\interior{(g \circ f)} \leq \interior{g} \circ \interior{f}$,
  required for $\interior{\FUN}$ to be a conucleus on $\HLL$, does not
  hold, unless $L$ is \cd.  The opposite inclusion
  $\interior{g} \circ \interior{f} \leq \interior{(g \circ f)}$ holds
  since $\interior{g} \circ \interior{f}$ belongs to $\HomR{L,L}$,
  $\interior{g} \circ \interior{f} \leq g \circ f$, and
  $\interior{(g \circ f)}$ is the greatest element of $\HomR{L,L}$
  below $g \circ f$. If
  $ \interior{(g \circ f)} = \interior{g} \circ \interior{f}$ for each
  $f,g \in \HLL$, then
  $\interior{1}$ is a unit for $\HRLL$ and $L$ is \cd.
  \end{inparaenum}
\end{proof}

%% file: dualizing.tex
\section{Dualizing elements of $\HLL$}

We investigate in this section dualizing elements of $\HLL$.
Proposition~\propno in \cite{EGHK2018} states that if
$\joinof{id_{L}}$ is dualizing, then $L$ is \cd.
Recall that a \emph{cyclic element} in a quantale $(Q,\circ)$ is an
element $0 \in Q$ such that $0 \upon x = x \below 0$, for each
$x \in Q$.  Trivially, the top element of a quantale is cyclic.
Our work \cite{2020-RAMICS} proves that if $\HLL$ has a non-trivial
cyclic element, then this element is $\joinof{id_{L}}$ and, once more,
cyclicity of $\joinof{id_{L}}$ implies that $L$ is \cd.  It was still
open the possibility that $\HLL$ might have dualizing elements and no
non-trivial cyclic elements. This is possible in principle, since the
tool Mace4 \cite{prover9-mace4} provided us with an example of a
quantale where the unique dualizing element is not cyclic. The
quantale is built on the modular lattice $M_{5}$ (with atoms
$u,d,a,b,c$) and has the following multiplication table:
$$
\begin{array}{r|ccccccc}
  & \bot & u  & d & a & b & c & \top \\
  \hline
  \bot&\bot  &  \bot &  \bot &  \bot &  \bot &  \bot &  \bot   \\
  u& \bot &  u &    d &  a &  b &  c & \top    \\
  d& \bot &  d  &  \top &  \top &  \top &  \top &  \top   \\
  a& \bot &  a  &  \top &  \top &  \top &  d &  \top   \\
  b& \bot &  b  &  \top &  d &  \top &  \top &  \top   \\
  c& \bot &  c  &  \top &  \top &  d &  \top &  \top  \\
  \top &  \bot &  \top &  \top &  \top &  \top &  \top &  \top   
\end{array}
$$
It is verified that $d$ is the only non-cyclic element and that, at
the same time, it is the only dualizing element. For the quantale
$\HLL$ we shall see that existence of a dualizing element again
implies \cdity of $L$ (and therefore existence of a cyclic and
dualizing element).

\medskip

As this might be of more general interest, we are going to investigate
how divisions $\intfun \below f$ and $f \upon \intfun$ by an arbitrary
$f \in \HLL$ act on the $e_{y,x}$ and $y \upTensor x$. To this end, we
start remarking that \Rt{s} intervene in the formulas for computing
\ladj{s} of the maps $c$ and $a$.
\begin{lemma}
  \label{prop:acAdj}
  \label{lemma:acAdj}
  The functions $c : Y \rto \HXY$ and
  $a : X^{op} \rto \HXY$ have both a left and a \radj. Namely, for 
  each $x\in X$, $y \in Y$, and $f \in \HXY$, the
  following relations hold:
  \begin{align*}
    c_{y} \leq f & \Tiff y \leq \meetof{f}(\bot)\,, &
    f \leq c_{y} & \Tiff f(\top) \leq y\,,\\
    a_{x} \leq f & \Tiff \Starf{\top} \leq x\,, & f \leq a_{x}& \Tiff
    x \leq \ra{f}(\bot)\,.
  \end{align*}
\end{lemma}
\begin{proofExtended}
  We have:
  \begin{enumerate}[(i)]
  \item $c_{y} \leq f$ iff $c_{y} \circ a_{\bot} \leq f$ iff
    $y \leq \meetof{f}(\bot)$.
  \item $f \leq c_{y}$ iff $f \leq y \tensor \top$ iff $f(\top)
    \leq y$.
    
  \item $a_{x} \leq f$ iff $c_{\top} \circ a_{x} \leq f$ iff $\top
    \leq \meetof{f}(x)$ iff $\la{\meetof{f}}(\top) \leq x$.
  \item $f \leq
    a_{x}$ iff $f \leq \bot \tensor x$ iff $f(x) \leq \bot$ iff $x
    \leq \ra{f}(\bot)$. \qedhere
  \end{enumerate}
\end{proofExtended}

Using the relations stated in Lemma~\ref{lemma:acAdj} computing divisions
becomes an easy task.
\begin{lemma}
  For each $x \in X$, $y \in Y$, and $f,g \in \HXY$, we have
  \begin{align*}
    f \upon a_{x} \ &  = \Perpf{x} \upTensor \top \quad (\,= c_{\Perpf{x}}\,)\,, &
    c_{y} \below f &  = \bot \upTensor \Starf{y} \quad (\,= a_{\Starf{y}}\,)\,, \\
    f \upon c_{y} & = \meetof{f}(\bot) \upTensor y\,, 
    & a_{x} \below f & = x \upTensor \Starf{\top}\,.
  \end{align*}
\end{lemma}
\begin{proof}
  We compute as follows:
  \begin{align*}
    g \leq c_{y} \below f & \Tiff c_{y} \circ g \leq f
    \Tiff c_{y} \circ a_{\bot} \circ g \leq f
    \Tiff c_{y}
    \circ a_{\ra{g}(\bot)} \leq f
    \Tiff 
    \Starf{y}  \leq \ra{g}(\bot) 
    \Tiff g \leq a_{\Starf{y}}\,.
  \end{align*}
  where we have used Lemma~\ref{lemma:relationDual}. 
  Verification that $f \upon a_{x} = c_{\Perpf{x}}$ is similar.
\begin{shortened}
     Dually:
    \begin{align*}
      g \leq f \upon a_{x} & \Tiff g \circ a_{x}
      \leq f \Tiff g \circ c_{\top} \circ a_{x} \leq f
      \Tiff c_{g(\top)}\circ a_{x} \leq f \\
      & \Tiff g(\top) \leq \meetof{f}(x) \Tiff g \leq
      c_{\meetof{f}(x)}\,.
    \end{align*}
\end{shortened}
The other two identities are verified as follows:
  \begin{align*}
    h \leq f \upon c_{y} & \Tiff h \circ c_{y} \leq f \Tiff c_{h(y)}
    \leq f \Tiff h(y) \leq \meetof{f}(\bot) \Tiff h \leq
    \meetof{f}(\bot) \upTensor y\,, \intertext{and, dually,}
    h \leq a_{x} \below f & \Tiff a_{x} \circ h \leq  f \Tiff a_{\ra{h}(x)} \leq f
    \Tiff \Starf{\top} \leq \ra{h}(x) 
    \Tiff h(\Starf{\top}) \leq x  
    \Tiff h \leq x \upTensor \Starf{\top}\,.
    \tag*{\qedhere}
  \end{align*}
\end{proof}
The relations in the following proposition are then easily derived.
\begin{proposition}
  \label{prop:divisions}
  The following relations hold:
  \begin{align*}
    f \upon (c_{y} \circ a_{x}) & = \Perpf{x} \upTensor y \,,
    & 
    (c_{y} \circ a_{x}) \below f & = x \upTensor \Starf{y}\,, \\
    f \upon (y \upTensor x) 
    & = \Perpf{x} \upTensor \top  \land \Perpf{\bot} \upTensor y\,,
    &
    (y \upTensor x) \below f
    & = \bot \upTensor \Starf{y} \land x \upTensor \Starf{\top}
    \,.
  \end{align*}
\end{proposition}
\begin{proofExtended}
  Compute as follows:
  \begin{align*}
    f \upon (c_{y} \circ a_{x}) & = (f \upon a_{x}) \upon c_{y}
    = c_{\Perpf{x}} \upon c_{y} = \Perpf{x} \upTensor y\,, \\
    (c_{y} \circ a_{x}) \below f & = a_{x} \below (c_{y} \below f) =
    a_{x} \below a_{(\Starf{y})} = x \upTensor \Starf{y}\,, \\
        f \upon (y \upTensor x) & = f \upon (c_{y} \vee a_{x}) = (f \upon
    c_{y}) \land (f \upon a_{x}) =  \Perpf{\bot} \upTensor y \land
    \Perpf{x} \upTensor \top \,,
    \\(y \upTensor x) \below f & = (c_{y} \vee a_{x})\below f =
    (c_{y}\below f) \land (a_{x} \below f) 
    = \bot \upTensor \Starf{y} \land (x \upTensor \Starf{\top}) 
   \,.
    \tag*{\qedhere}
  \end{align*}
\end{proofExtended}
  From Proposition~\ref{prop:divisions}, it follows that
  \begin{align*}
    c_{\Perpf{x}} \circ a_{y} & \leq f \upon (y \upTensor x) \,,
    &
    c_{x} \circ a_{\Starf{y}} & \leq (y \upTensor x) \below f
    \,. 
  \end{align*}

  If $\Star{f}$ is invertible, then its inverse is also its
  \radj. From the uniqueness of the right adjoint it follows that
  $\Star{f}$ is inverted by its \radj $\meetof{f}$. Thus, $\Star{f}$
  is invertible if and only if $\meetof{f}$ is invertible, in which
  case we have $\Star{f}(\top) = \top$ and $\meetof{f}(\bot) = \bot$,
  since both $\Star{f}$ and $\meetof{f}$ are bicontinuous.
In some important case, the expressions exhibited in
Proposition~\ref{prop:divisions} simplify:
\begin{corollary}
  \label{corollary:transformTensors}
  If $\meetof{f}(\bot) = \bot$ (resp., $\Star{f}(\top) = \top$), then
  \begin{align*}
    f \upon (y \upTensor x) 
    &= c_{\Perpf{x}} \circ a_{y}
    \quad (\text{resp.,} \;\;
    (y \upTensor x) \below f
    = c_{x} \circ a_{\Starf{y}})\,.
  \end{align*}
  These relations hold as soon as either $\meetof{f}$ or $\Star{f}$ is
  invertible.
\end{corollary}

\begin{example}
  If $L$ is a \cdlatt, then the relation $\jm{id_{L}} = id_{L}$ holds,
  by \RT \cite{Raney60}. Let $o = \joinof{id_{L}}$, then $\meetof{o}$
  is invertible, since it is the identity.  Necessarily, we also have
  $\Star{o} = id_{L}$ and therefore: 
  \begin{align*}
    o \upon (c_{y} \circ a_{x}) & = (c_{y} \circ a_{x}) \below o = x
    \upTensor y \,,
    &
    o \upon (y \upTensor x) & = (y \upTensor x) \below o  =
    c_{x} \circ a_{y}\,.
    \tag*{\EndOfExample}
  \end{align*}
\end{example}

\begin{theorem}
  \label{thm:bijection}
  If $f \in \EL$ is dualizing, then $\meetof{f}$ and
  $\Star{f}$
  are inverse to each other and $L$ is \cd.
\end{theorem}
\begin{proof}
  If $f$ is dualizing, then, for all $x,y \in L$,
  \begin{align*}
    c_{y} & = f \upon (c_{y} \below f) = f \upon a_{\Starf{y}} =
    c_{\Perpf{\Starf{y}}} \,,
    \;\quad
    a_{x}  = (f \upon a_{x}) \below f = c_{\Perpf{x}} \below f
    = a_{\Starf{\Perpf{x}}} \,,
  \end{align*}
  and since both $c$ and $a$ are injective, then
  $y = \Perpf{\Starf{y}}$ and $x = \Starf{\Perpf{x}}$. Thus,
  $\meetof{f}$ and $\Star{f}$ are inverse to each other.
  Remark now that $\Star{f} \in \HomR{L,L}$, since
  $\Star{f} = \la{\meetof{f}} = \joinof{\ra{f}}$.  Then
  $id_{L} = \meetof{f} \circ \Star{f} \in \HomR{L,L}$, since
  $\HomR{L,L}$ is an ideal of $\EL$. Then $L$ is \cd by \RT.
\end{proof}
It is not difficult to give a direct proof of the converse, namely
that if $\Star{f}$ is invertible, then $f$ is dualizing. We prove this
as a general statement about \irl{s}. Notice that the map sending $f$
to $\Star{f}$ is definable in the language of \irl{s}, since
$\Star{f} = f \below o$, where $o = \joinof{id_{L}}$ is the canonical
cyclic dualizing element of $\EL$ (if $L$ is \cd). 
The statement in the following Proposition~\ref{prop:bijection} is
implicit in the definition of a (symmetric) compact closed category in
\cite{KellyLaplaza1980} (see \cite[\S 5]{Yetter2001} for the non
symmetric version of this notion).
\begin{proposition}
  \label{prop:bijection}
  In every \irl $Q$, $f$ is dualizing if and only if $\Star{f}$ is
  invertible.
\end{proposition}
\begin{proof}
  If $f \in Q$ is dualizing, then
  $f \upon (\Star{x} \below f) = (f \upon \Star{x}) \below f =
  \Star{x}$, for each $x \in Q$. Using well known identities of
  \irl{s}, compute as follows:
  \begin{align*}
    x & = \Star{\Star{x}} = \Star{(f \upon (\Star{x} \below f))} = (\Star{x}
    \below f) \circ \Star{f} = (x \upon \Star{f}) \circ \Star{f}\,.
  \end{align*}
  Letting $x = 1$, then $1 = (1 \upon \Star{f}) \circ \Star{f}$.
  We derive $1 = \Star{f} \circ (\Star{f} \below 1)$ similarly, from
  which it follows that $\Star{f}$ is inverted by $1 \upon \Star{f} =
  \Star{f} \below 1$.

  Conversely, suppose that $\Star{f}$ is invertible, say
  $\Star{f} \circ g = g \circ \Star{f} = 1$. It immediately follows
  that $g = 1 \upon \Star{f} = \Star{f} \below 1$. Again, for each
  $x \in Q$,
  \begin{align*}
    \Star{x} & = \Star{x} \circ (1 \upon \Star{f}) \circ \Star{f} \leq (\Star{x}
    \upon \Star{f}) \circ \Star{f} \,,
  \end{align*}
  and then, dualizing this relation, we obtain
  \begin{align*}
    x = \Star{\Star{x}} & \geq \Star{((\Star{x}
    \upon \Star{f}) \circ \Star{f})} =  f \upon (x \below f)\,.
  \end{align*}
  Since $x \leq f \upon (x \below f)$ always hold, we have
  $x = f \upon (x \below f)$. The identity $x = (f \upon x) \below f$
  is derived similarly.
 \end{proof}

\begin{example}
  If $L = [0,1]$, then $f$ is dualizing if and only if it is
  invertible. Indeed, $f$ is dualizing iff $\Star{f}$ is invertible
  iff $\meetof{f}$ is invertible.  Now, if $\meetof{f}$ is invertible,
  then it is continuous and $\meetof{f} = f$; therefore $f$ is
  invertible.  Similarly, if $f$ is invertible, then it is continuous
  and $\meetof{f} = f$; therefore $\meetof{f}$ is invertible.
  \eExample
\end{example}

\begin{example}
  Consider a poset $P$, the \clatt $\DP$ of downsets of $P$, and
  recall that $\DP$ is \cd. 
  The quantale $\HomJ{\DP,\DP}$ is isomorphic to the quantale of
  weakening relations (profuctors/bimodules) on the poset $P$. These
  are the relations $R \subseteq P \times P$ such that $yRx$,
  $y' \leq y$, and $x\leq x'$ imply $x'R y'$ (for all
  $x,x',y,y' \in P$). Thus, weakening relations are downsets of
  $P \times P^{op}$ and the bijection between $\HomJ{\DP,\DP}$ and
  $\DP[P \times P^{op}]$ goes along the lines described in previous
  sections, since
  \begin{align*}
    \HomJ{\DP,\DP} & \simeq \DP \tensor \DP[P^{op}] \simeq \DP[P \times P^{op}]\,. 
  \end{align*}
  Explicitly, this bijection, sending $f$ to $R_{f}$, is such that
  \begin{align}
    \label{eq:bijectionR}
    (y,x) \in R_{f} & \Tiff y \in f(\dset x) \Tiff c_{\dset y} \circ
    a_{\dset x} \leq f\,,
  \end{align}
  where, for $x \in P$, $\dset x \eqdef \set{y \in P \mid y \leq x}$.
  We have seen that dualizing elements of $\HomJ{\DP,\DP}$ are in
  bijection with automorphisms of $\DP$ which in turn are in bijection
  with automorphisms of $P$.  Given such an automorphism, we aim at
  computing the dualizing weakening relation corresponding to this
  automorphism. To this end, let us recall that, when $L$ is \cd,
  $o = \joinof{id} = \ell(\meetof{id})$ is the unique non-trivial
  cyclic element. Observe that since $o = \Star{id}$, $o$ is also
  dualizing and that
  \begin{align*}
    (y,x) \in R_{o} &\Tiff x \not\leq y\,.
  \end{align*}
  For $f \in \HomJ{\D(P),\D(P)}$, recall that $\Star{f} = f \below
  o$. We use the relations in \eqref{eq:bijectionR} to compute the
  dualizing element of $\DP[P\times P^{op}]$ corresponding to an
  invertible order preserving map $g : P \rto P$, as follows:
  \begin{align*}
    (y,x) \in R_{\Star{\DP[g]}}
    & \Tiff c_{\dset y}\circ a_{\dset x} \leq \Star{\DP[g]} = \DP[g] \below
    o \\
    & \Tiff \DP[g] \circ c_{\dset y}\circ a_{\dset x}  =  
    c_{\DP[g](\dset y)}\circ a_{\dset x}
    = c_{\dset g(y)}\circ a_{\dset x} \leq 0 \\
    &
    \Tiff
    x \not \leq g(y)\,.
    \tag*{\eExample}
  \end{align*}
\end{example}

\begin{shortened}

\begin{theorem}
  \label{thm:bijection}
  If $f \in \EL$ is dualizing if and only if $\meetof{f}$ and
  $\Star{f}$
  are inverse to each other and $L$ is \cd.
\end{theorem}
\begin{proof}
  If $f$ is dualizing, then, for all $x,y \in L$,
  \begin{align*}
    c_{y} & = f \upon (c_{y} \below f) = f \upon a_{\Starf{y}} =
    c_{\Perpf{\Starf{y}}} \,,
    \;
    a_{x}  = (f \upon a_{x}) \below f = c_{\Perpf{x}} \below f
    = a_{\Starf{\Perpf{x}}} \,,
  \end{align*}
  and since both $c$ and $a$ are injective, then
  $y = \Perpf{\Starf{y}}$ and $x = \Starf{\Perpf{x}}$. Thus,
  $\meetof{f}$ and $\Star{f}$ are inverse to each other.
    
  Suppose now that $\meetof{f}$ is invertible. 
  For an arbitrary $g \in \EL$, we need to argue that
  $g = f \upon (g \below f) = (f \upon g) \below f$.  Since the
  inclusions $g \leq f \upon (g\below f)$ and
  $g \leq (f \upon g)\below f$ always hold, we only need to argue that
  $f \upon (g\below f) \leq g$ and $(f \upon g)\below f \leq g$.
  Write $g = \bigwedge_{g(x) \leq y} y \upTensor x$.
  Since $\intfun \below f$ is antitone, and using the relations
  pointed out in Remark~\ref{remark:SemiTransformTensors}, the
  following inequalities hold:
  \begin{align*}
    g \below f & = (\bigwedge_{g(x) \leq y} y \upTensor x) \below f \geq
    \bigvee_{g(x) \leq y} ((y \upTensor x) \below f)
    \geq \bigvee_{g(x) \leq y} c_{x} \circ a_{\Starf{y}}\,.
  \end{align*}
  Since $f \upon \intfun $ is antitone and moreover it transforms
  joins into meets, we deduce
  \begin{align*}
    f \upon (g\below f) & \leq f \upon (\bigvee_{g(x) \leq y} c_{x} \circ
    a_{\Starf{y}}) =
     \bigwedge_{g(x) \leq y} 
    f \upon (c_{x} \circ
    a_{\Starf{y}})
    \\
    &
    =  \bigwedge_{g(x) \leq y}  \Perpf{\Starf{y}} \upTensor  x
    =  \bigwedge_{g(x) \leq y} y \upTensor x = g\,,
  \end{align*}
  where we have used $\meetof{f} \circ \Star{f} = id_{L}$.
  We deduce $(f \upon g)\below f \leq g$ in a similar way,
  using this time $ \Star{f} \circ \meetof{f}  = id_{L}$.
\end{proof}

  \begin{theorem}
    If $\EL$ has a dualizing element, then $L$ is \cd.
  \end{theorem}
  \begin{proof}
    There are many way of observing this.

    1. This is essentially the same argument developed in
    \cite{EGHK2018} when $f = \joinof{id}$.  Our previous computations
    show that if $f$ is dualizing, then $\intfun\below f$ sends the
    \et $y \upTensor x$ to $c_{x} \circ a_{\Starf{x}}$.
    As a shortcut towards establishing this relation, we recall that
    if $f$ is a dualizing element of a quantale $(Q,\circ)$, then the
    relation $(g \upon h) \below f = h \circ (g \below f) $ holds, for
    any $g,h \in Q$.  Thus (still using some of our observations)
    $(y \upTensor x) \below f = (c_{y} \upon c_{x}) \below f = c_{x}
    \circ (c_{y} \below f) = c_{x} \circ a_{\Starf{x}}$.

    Since the \et{s} $y \upTensor x$ generate $\EL$ under meets, then
    the $c_{x} \circ a_{y}$ generate $\EL$ under joins. Therefore,
    $\HomJ{L,L} = \HomR{L,L}$ and $L$ is \cd by \RT.

    2. Remark that $\Star{f} \in \HomR{L,L}$, since
    $\Star{f} = \joinof{\ra{f}}$.  If $f$ is dualizing, then
    $\Star{f}$ is inverted by $\meetof{f}$, and therefore
    $id_{L} = \meetof{f} \circ \Star{f} \in \HomR{L,L}$, since this is
    an ideal of $\EL$.
    Therefore $id_{L} = \jm{id_{L}}$ and, again, $L$ is \cd by \RT.
  \end{proof}
\end{shortened}

%% file: further.tex
\section{Further \bimorphism{s}, bijections, directions}

Even when \Rt{s} are not inverse to each other, it might still be
asked whether there are other isomorphisms between $\HomM{L,L}$ and
$\EL$. By Fact~\ref{fact:bijection}, this question amounts to
understand whether $\EL$ is autodual.

Let us discuss the case when $L$ is a finite lattice.  We use $\JL$
for the set of \jirr elements of $L$ and $\ML$ for the set of \mirr
elements of $L$.
The reader will have no difficulties convincing himself of the
following statement:
\begin{lemma}
  A map $f \in \EL$ is \mirr if and only if it is an elementary
  tensor of the form $m \upTensor j$ with $m \in \ML$ and $j \in \JL$.
\end{lemma}
\begin{proofExtended}
  Every element of $\EL$ is a meet of elementary tensors and therefore
  \mirr element is an elementary tensor.  If $x,y \in L$,
  $y = \bigwedge_{i \in I} m_{i}$ and $x = \bigvee_{k \in K} j_{k}$
  with $m_{i} \in \ML$ and $j_{k }\in \JL$ for each $i \in I $ and
  $k \in K$, then
  $y \upTensor x = \bigwedge_{i \in I,k \in K} m_{i} \upTensor j_{k}$.
  Thus, every \mirr element of $\EL$ is an elementary tensor of the
  form $m \otimes j$ with $m \in \ML$ and $j \in \JL$.

  Let us argue that every such elementary tensor is \mirr.  We shall
  argue that $m^{\ast}\upTensor j \land m \upTensor j_{\ast}$ is the
  unique upper cover of $m \upTensor j$ and, to this goal, it shall be
  is enough to argue that if $y \upTensor x$ is an elementary tensor
  with $m \upTensor j \leq y \upTensor x$, then
  $m^{\ast}\upTensor j \leq y \upTensor x$ or
  $m \upTensor j_{\ast} \leq y \upTensor x$.

  The relation $m \upTensor j \leq y \upTensor x$ yields $m \leq y$ and
  $x \leq j$, while the relation $y \upTensor x \not\leq m \upTensor j$
  yields $(y \upTensor x)(j) \not\leq m$. If $x < j$, then
  $m \upTensor j_{\ast} \leq y \upTensor x$.  Thus suppose that $x = j$,
  so $y = (y \upTensor x)(j) \not\leq m$. Thus $m < y$,
  $m^{\ast} \leq y$ and $m^{\ast} \upTensor j \leq y \upTensor x$.
\end{proofExtended}
The following statement might instead be less immediate:
\begin{lemma}
  For each $j \in \JL$ and $m \in \ML$, the map $e_{j,m}$ is \jirr.
\end{lemma}
\begin{proof}
  Let $m \in \ML$ and $j \in \JL$, and let us use $m^{\ast}$ to denote
  the unique upper cover of $m$. Suppose that
  $e_{j,m} = \bigvee_{i \in I} f_{i}$.  By evaluating the two sides of
  this equality at $m^{\ast}$, we obtain
  $j = \bigvee_{i \in I} f_{i}(m^{\ast})$ and therefore
  $j = f_{i}(m^{\ast})$ for some $i \in I$.
  If $t \leq m$, then $f_{i}(t) \leq e_{m,j}(t) = \bot$. Suppose now
  that $t \not \leq m$, so $m < t \vee m$ and
  $m^{\ast} \leq t \vee m$.  Observe also that
  $f_{i}(t) \leq e_{j,m}(t) = j$, since
  $e_{j,m} = \bigvee_{i \in I} f_{i}$.  Then
  $j = f_{i}(m^{\ast}) \leq f_{i}(m \vee t) = f_{i}(m) \vee f_{i}(t) =
  \bot \vee f_{i}(t) = f_{i}(t)$, so $j = f_{i}(t)$. We have argued
  that $f_{i}(t) = e_{j,m}$, for all $t \in L$, and therefore that
  $f_{i} = e_{j,m}$.
\end{proof}

\begin{theorem}
  If $L$ is a finite lattice and  $\EL$ is autodual, then
  $L$ is distributive.
\end{theorem}
\begin{proof}
  If $\psi : \EL^{op} \rto \EL$ is invertible, then $\psi$ restricts
  to a bijection $\M(\EL) \rto \J(\EL)$, so these two sets have same
  cardinality.
  For $m \in \ML$ and $j \in \JL$, the $e_{j,m}$ as well as the
  elementary tensors $m \upTensor j$ are pairwise distinct.
  Therefore, we have
  $$
  \card{\ML} \times \card{\JL} \leq \card{\J(\EL)} = \card{\M(\EL)} =
  \card{\ML} \times \card{\JL}
  $$ 
  and $\card{\ML} \times \card{\JL} = \card{\J(\EL)}$.  That is, the
  elements $e_{j,m}$ are all the \jirr elements of $\EL$ and therefore
  the set $\set{e_{j,m} \mid j \in \JL, m \in \ML }$ generates
  $\EL$ under joins. It follows that $\EL = \HomR{L,L}$ and that $L$
  is distributive.
\end{proof}

We do not know yet if the theorem above can be generalized to infinite
\clatt{s} or whether there is some fancy infinite \clatt $L$ that is
not \cd and such that $\EL$ is autodual. It is clear, however, that in
order to construct such a fancy lattice, properties of \bimorphism{s}
$\psi : L \times L^{op} \rto \EL$ need to be investigated. What are
the properties of a \bimorphism $\psi$ forcing $L$ to be \cd when
$\tilde{\psi}$ is surjective? Taking the \bimorphism $e$ as example,
let us abstract part of \RT:
\begin{proposition}
  \label{prop:Rabstract}
  Let $\psi: L \times L^{op} \rto \EL$ be a \bimorphism such that for
  each $x,y \in L$, the image of $L$ under $\psi(y,x)$ is a finite
  chain.  If $id_{L}$ belongs to the image of
  $\tilde{\psi} : \HomM{L,L} \rto \HomJ{L,L}$, then $L$ is a \cdlatt.
\end{proposition}
\begin{proof}
  Let $z \in Z$ such that
  $z = \bigwedge_{i \in I} \bigvee_{j \in J_{i}} z_{j}$.  We aim at
  showing that $z \leq \bigvee_{s} \bigwedge_{i \in I} z_{s(i)}$, with
  the index $s$ ranging on choice functions
  $s : I \rto \bigcup_{i \in I} J_{i}$ ($s$ is a choice function if
  $s(i) \in J_{i}$, for each $i \in I$).
  Since $id_{L} = \bigvee \set{\psi(y,x) \mid \psi(y,x) \leq id_{L}}$,
  we also have
  $z = \bigvee \set{\psi(y,x)(z) \mid \psi(y,x) \leq id_{L}}$ and
  therefore, in order to achieve our goal, it will be enough to show
  that for each $y,x \in L$, if $\psi(y,x) \leq id_{L}$, then
  $\psi(y,x)(z) \leq \bigvee_{s} \bigwedge_{i \in I} z_{s(i)}$.
  Let $y,x$ be such that $\psi(y,x) \leq id_{L}$, fix $i \in I$, and
  observe then that
  \begin{align*}
    \psi(y,x)(z) & \leq \psi(y,x)(\bigvee_{j \in J_{i}} z_{j}) =
    \bigvee_{j \in J_{i}} \psi(y,x)(z_{j})\,,
  \end{align*}
  since $z \leq \bigvee_{j \in J_{i}} z_{j}$.  Since the set
  $\set{ \psi(y,x)(z_{j}) \mid j \in J_{i}}$ is finite and directed
  (it is a finite chain), it has a maximum: there exists
  $j(i) \in J_{j}$ such that
  $\bigvee_{j \in J_{i}} \psi(y,x)(z_{j}) = \psi(y,x)(z_{j(i)})$.  It
  follows that
  \begin{align*}
    \psi(y,x)(z) & \leq \psi(y,x)(z_{j(i)}) \leq z_{j(i)}\,,
  \end{align*}
  since $\psi(y,x) \leq id_{L}$.  By letting $i$ vary, we have
  constructed a choice function $j : I \rto \bigcup_{i \in I} J_{i}$
  such that $\psi(y,x)(z) \leq \bigwedge_{i \in I} z_{j(i)}$, and consequently
  $\psi(y,x)(z) \leq \bigvee_{s} \bigwedge_{i \in I} z_{s(i)}$.
\end{proof}
Bimorphisms satisfying the conditions of
Proposition~\ref{prop:Rabstract} might be easily constructed by taking
$f \in \HomJ{ L, \EL}$,  $g \in \HomJ{L^{op},L^{op}}$ (resp.,
$f \in \EL$ and $g \in \HomJ{L^{op},\HomJ{L,L}}$), and defining then
\begin{align*}
  \psi(y,x) & \eqdef f(y) \circ a_{g(x)}
  \quad
  (\, \text{resp., } \psi(y,x) \eqdef c_{f(y)} \circ g(x) \,)\,.
\end{align*}
These \bimorphism{s}  satisfy
the conditions of Proposition~\ref{prop:Rabstract}, since they only
take two values.  As a consequence of the proposition, they cannot be
used to construct a fancy dual isomorphism of $\HLL$.